\documentclass[amsmath,amssymb]{revtex4-1} 
\usepackage{amsthm}

\theoremstyle{plain}
\newtheorem{thm}{\protect\theoremname}
\theoremstyle{plain}
\newtheorem{prop}[thm]{\protect\propositionname}
\theoremstyle{plain}
\newtheorem{lem}[thm]{\protect\lemmaname}
\theoremstyle{plain}
\newtheorem{cor}[thm]{\protect\corollaryname}
\theoremstyle{plain}
\newtheorem*{prop*}{\protect\propositionname}
\theoremstyle{plain}
\newtheorem*{lem*}{\protect\lemmaname}

\makeatother

\providecommand{\corollaryname}{Corollary}
\providecommand{\lemmaname}{Lemma}
\providecommand{\propositionname}{Proposition}
\providecommand{\theoremname}{Theorem}
\newcommand{\dx}{{dx\over(2\pi)^3}}

\begin{document}

\title{An Infinitesimal Quantum Group Underlies Classical Fluid Mechanics} 
\author{S. G. Rajeev}
\affiliation{
Department of Physics and Astronomy\\
Department of Mathematics\\ 
University of Rochester\\
Rochester,NY 14618,USA}
\begin{abstract}

Arnold showed that the Euler equations of an ideal fluid describe geodesics in the Lie
algebra of incompressible vector fields.  We will show that 
 helicity induces a splitting of the Lie algebra into two isotropic subspaces, forming a 
 Manin triple. Viewed another way, this shows that there is an infinitesimal quantum 
 group (a.k.a. Lie bi-algebra) underlying classical fluid mechanics.

\end{abstract} 

\maketitle

Infinite Dimensional Lie algebras are important in physics as symmetries of
 field theories, such as the standard model of particle physics. The best understood 
 cases are $1+1$ dimensional theories, such as the Kac--Moody\cite{Kac} and  Virasoro algebras. The latter also arises in fluid mechanics in one spatial 
dimension\cite{RajeevFluid}: being essentially the Lie algebra of vector fields on a circle it is the phase space of a  fluid with periodic boundary conditions. By now, its geometry\cite{BowickRajeev},representation theory\cite{Kac}  and physical meaning are all quite understood. The correlation functions of many conformal field theories (and related critical statistical models) can be  found exactly using the representation  theory\cite{FQS, BPZ}
  of the Virasoro and Kac-Moody  algebras.

A similar theory of Lie algebras that arise from $3+1$ dimensional field theories would
be of great interest,as they would apply to more realistic physical systems.  A good candidate is the Lie algebra $\mathcal{S}$ of incompressible (i.e.,satisfying $\nabla\cdot u=0$ ) vector fields. Arnold\cite{Arnold} showed that the Euler equations describe geodesics on it; its  representations (more generally ``quantum" deformations) are likely to be useful in determining the correlation functions of velocity in a turbulent fluid.  There are many analogies with the Virasoro and Kac-Moody algebras: $\mathcal{S}$ is a graded Lie algebra, with a non-trivial central extension
$\hat{\mathcal{S}}$. It is known that $\mathcal{S}$ admits an invariant inner product
(related to helicity in fluid mechanics). We will show in this paper that it can be
extended to $\hat{\mathcal{S}}$, as for Kac-Moody algebras. Moreover, we will show that
$\hat{\mathcal{S}}$ is a Lie bi-algebra. This opens up an as
yet unsolved problem: exponentiate $\hat{\mathcal{S}}$ to a ``quantum group" (more
precisely, a Hopf algebra). This could generalize  to three dimensions of the symplectic integration methods available in two dimensions\cite{McLachlan}.

The commutator $[u,w]=u\cdot \nabla w-w\cdot\nabla u$ is anti-symmetric and satisfies the
Jacobi identity.
The identity $[u,w]=\nabla\times\left(w\times u\right)+w\nabla.u-u\nabla.w$ shows that
the commutator of
incompressible vector fields is again incompressible; i.e., they form a Lie algebra
$\mathcal{S}$. It shows
a bit more: that the commutator is  exact; i.e., the curl of some vector field. The
exact vector fields
form a sub-algebra (indeed an ideal) $\mathcal{S}'\subset \mathcal{S}$.

We will impose periodic boundary conditions; i.e., space will be a torus $\mathbb{T}^3$
of side $2\pi$. A
simple argument using Fourier analysis shows that, on $\mathbb{T}^3 $, any vector field
with zero average is
exact (i.e., the curl operator is invertible in $\mathcal{S}'$). So, any vector field can
be written as the
sum of a constant vector field (its average $\bar{u}=\int u{dx\over (2\pi)^3}$) and an
exact vector field
$u'=u-\bar{u}$. Since constant vector fields commute, we get a semi-direct product
$\mathcal{S}=\mathbb{R}^3\ltimes \mathcal{S}'$.

$\mathcal{S}'$ has an invariant inner product $\langle u',w'\rangle=\int
u'.\mathrm{curl}^{-1}w'{d^{3}x\over
(2\pi)^3}$ .Symmetry follows by integration by parts. To see the invariance, use the
commutator identity
above to write $\langle u',[v',w']\rangle=\int u'.(v'\times w'){d^{3}x\over (2\pi)^3}$
and use the
anti-symmetry of the triple scalar product to get $
\langle\left[v',u'\right],w']\rangle+\langle
u',[v',w']\rangle=0$. It is annoying that this inner product does not extend to the full
algebra
$\mathcal{S}$, because curl is not invertible there.

There is a central extension $\mathbb{R}^3\to \hat{\mathcal{S}}\to\mathcal{S}$
(discovered in a different
context\cite{CentralExtension} ): \[
\left[(u,\eta),(w,\mu)\right]=\left([u,w],\Omega(u,w)\right) \] with
co-cycle $\Omega(u,w)=\int w\times u{d^{3}x\over (2\pi)^3}$. By pairing the constant
vectors with the
translations (by analogy with the Kac-Moody algebra\cite{Kac}), we can get an invariant
inner product on the
extended Lie algebra $\hat{\mathcal{S}}$: \[ \langle(u,\eta),(w,\mu)\rangle= \int
u'.\mathrm{curl}^{-1}w'{d^{3}x\over (2\pi)^3} +\eta. \bar{w}+\mu.\bar{u}. \]

This central extension serves a useful physical purpose: it allows us to extend the
geodesic interpretation
of Euler's equation to include translations. Recall\cite{RajeevMechanics} that the Euler equations of an
incompressible ideal
fluid\[ \frac{\partial v}{\partial t}+v.\nabla v=-\nabla p,\quad\nabla.v=0 \] can be
written also as
$\frac{\partial v}{\partial t}+\omega\times v=-\nabla\left(p+\frac{1}{2}v^{2}\right),$
where
$\omega=\mathrm{curl}\ v$ is the vorticity. By taking the curl we get the vorticity form
of these equations
\[ \frac{\partial\omega}{\partial t}+[v,\omega]=0,\quad\omega=\nabla\times v \] But this
is incomplete: we
also need to know how the constant part of $v$ (which is lost in $\nabla\times v$)
evolves in time. By
averaging Euler's equation we also get
\[
 \frac{\partial\bar{v}}{\partial t}+\int\omega\times v {d^{3}x\over
(2\pi)^3}=0. 
\] This suggests that we should regard the $\mathrm{curl}$ operator as
extended to the central
extension by $\mathrm{curl}:\mathcal{S}\to\hat{\mathcal{S}'}$, \[
\mathrm{curl}v=\left(\nabla\times v,\int v
{d^{3}x\over (2\pi)^3}\right)\equiv\hat{\omega}, \] So velocity belongs to $\mathcal{S}$
(i.e., includes
constant vectors) . But vorticity belongs to $\hat{\mathcal{S}'}$: total vorticity $\int
\omega {dx\over
(2\pi)^3}$ is zero, but extended vorticity $\hat{\omega}$ has a central component equal
to average velocity.
The extended notion of curl above is an invertible operator.

The Euler equation now becomes \[ \frac{\partial\hat{\omega}}{\partial
t}+\left[v,\hat{\omega}\right]=0 \]
We can verify that \[ \langle\hat{\omega},\mathrm{curl}^{-1}\hat{\omega}\rangle=\int
v^{2}{dx\over (2\pi)^3}
\]
allowing us to interpret the total kinetic energy as
$H=\frac{(2\pi)^3}{2}\langle\hat{\omega},\mathrm{curl}^{-1}\hat{\omega}\rangle$.

There is a close analogy to the rigid body: $\hat{\omega}$ is like the angular momentum
$L$ and
$\mathrm{curl}$ is analogous to moment of inertia. Moreover,
$\langle\hat{\omega},\hat{\omega}\rangle$ is
analogous to $L\cdot L$; being invariant under the Lie algebra, it is a conserved quantity.
In fluid
mechanics,this is called helicity\cite{ArnoldKhesin,Helicity}.

Thus, fluid mechanics should be thought of as geodesic motion in the \emph{extended} Lie
algebra
$\mathcal{S}'$. The central elements are trivially time-independent. It is amusing that
total momentum of
the fluid $\int vd^{3}x$ is not conserved: it has an ``anomaly'' proportional  to the
cocycle $\Omega(v,\omega)$.

Since parity changes its sign, the invariant inner product $\langle,\rangle$ is not
positive. This raises
the possibility that there is a splitting $\hat{\mathcal{S}}=\mathcal{A}\oplus\mathcal{B}$
into isotropic
subalgebras; i.e., such that $\langle
\mathcal{A},\mathcal{A}\rangle=0=\langle\mathcal{B},\mathcal{B}\rangle$. The meaning of
such a ``Manin
triple"\cite{Ogievetsky,QuantumGroupBooks} $(\hat{\mathcal{S}},\mathcal{A},\mathcal{B}$)
is clearer in a
basis $X_a\in \mathcal{A}$ and $X^a\in\mathcal{B}$ satisfying $\langle
X_a,X_b\rangle=0=\langle
X^a,X^b\rangle,\langle X_a,X^b\rangle=\delta_a^b$ and \[
[X_{a},X_{b}]=\Gamma_{ab}^{c}X_{c},\quad
[X^{a},X^{b}]=\mu_{c}^{ab}X_{c} \] \[
[X_{a},X^{b}]=-\Gamma_{ad}^{b}X^{d}+\mu_{a}^{bd}X_{d} \] These
relations define a Lie bi-algebra\cite{QuantumGroupBooks,Ogievetsky};i.e., infinitesimal
versions of quantum
groups (Hopf algebras) which are also called Poisson-Hopf algebras. It is expected that
every Lie bi-algebra
(even infinite dimensional ones) are infinitesimal versions of Hopf algebras. There is a
general procedure
for exponentiating them in finite dimensions; but not yet in the general case.

To construct such a basis for $\hat{\mathcal{S}}$ we will use Fourier analysis and an
elementary trick using irrational numbers.  Let $\alpha$, $\beta$ and $\gamma\equiv\alpha\times\beta$ be constant
vectors satisfying the
conditions that for $m \in\mathbb{Z}'^{3}$, \[ \alpha\cdot
m\neq0,\quad\beta.m\neq0,\quad(\alpha\times\beta)\cdot m\neq0,. \] (Here
$\mathbb{Z}'^{3}$is the set of
non-zero 3-vectors with integer components). This ``no resonance'' condition implies that
\[ \alpha\times
m\neq0,\quad\beta\times m\neq0,\quad\gamma\times m\neq0 \] as well. For example, \[
\alpha=\left\{
1,\sqrt{2},\sqrt{3}\right\} ,\quad\beta=\left\{ \sqrt{3},1,\sqrt{2}\right\} \] leading to
$\gamma=\left\{2-\sqrt{3},3-\sqrt{2},1-\sqrt{6}\right\}$. (There are many other choices
also, leading to equivalent bases.) The orbits of $\alpha,\beta,\gamma$ are dense in $\mathbb{T}^3$, so that the only continuous functions that vanish on them are zero everywhere. This allows  us to invert operators such as $\mathrm{curl}$ and $\gamma\cdot\nabla$ subject to periodic boundary conditions. Such ideas occur also in the proof of the KAM theorem of classical mechanics\cite{RajeevMechanics}.

It is not difficult to verify that for $m\neq0$, we can expand any vector satisfying
$m.w=0$ as $
w=\frac{m\cdot w\ \beta.m-m\cdot m\beta.w}{m\cdot m\ \gamma\cdot m}\alpha\times
m+\frac{-m\cdot w\
\alpha.m+m\cdot m\alpha.w}{m\cdot m\ \gamma\cdot m}\beta\times w$. Thus

\[ a_{m}=-\alpha\times\nabla e_{m},\quad b_{m}=-\beta\times\nabla e_{m} , \] (where
$e_m=e^{im\cdot x}$)
form a basis for $\mathcal{S}'$. To extend the basis to $\hat{\mathcal{S}}$ we add three
central elements
$c_j$
and the three translations along the co-ordinate axes $d_j$. We can calculate
\begin{widetext}
\begin{equation} [a_{m},a_{n}]=\alpha\cdot(m\times
n)a_{m+n},\quad[b_{m},b_{n}]=\beta\cdot(m\times
n)b_{m+n}\label{eq:aabb} \end{equation} \[
[a_{m},b_{n}]=\frac{\gamma.m}{\gamma\cdot(m+n)}\beta\cdot(m\times
n)a_{m+n}+\frac{\gamma\cdot n}{\gamma\cdot(m+n)}\alpha\cdot(m\times
n)b_{m+n}\label{ab}+\gamma\cdot n\
c.n\delta(m+n) \] \end{widetext} In addition,$[*,c_i]=0$ and $[d_j,d_k]=0$ and \[
\left[d_j,a_{m}\right]=im_ja_{m},\quad\left[d_j,b_{m}\right]=im_j b_{m}. \]

The invariant inner product becomes, 
\[ 
\langle a_{m},a_{n}\rangle=\langle a_m,c_j\rangle=0=\langle b_{m},b_{n}\rangle=
\langle b_m , d_j\rangle
 \]
\begin{equation}
\langle a_{m},b_{n}\rangle=i\gamma\cdot n\delta_{m+n},\quad \langle
d_j,c_k\rangle=\delta_{jk}\label{eq:IP} 
\end{equation}

We can change the basis slightly to make the Lie bi-algebra structure more obvious: \[
X_{m}=a_{m},\quad
X^{n}=-\frac{1}{i\gamma.n}b_{-n} \] so that $\langle X_{m},X^{n}\rangle=\delta_{m}^{n}$.
The commutation
relations in this basis are of the required type with $
\Gamma_{mn}^{k}\equiv\delta_{m+n}^{k}\alpha.(m\times
n),
\mu_{k}^{mn}\equiv\delta_{k}^{m+n}\frac{(-i\gamma.k)}{(i\gamma.m)(i\gamma.n)}\beta.(m\times
n). $  Then $\mathcal{A}$ is spanned by $X_a=(X_m,c_j)$ and $\mathcal{B}$ by $X^a=(X^m,d_j)$ .

The $L^2$ inner product is
 \[ (a_m,a_n)=(\alpha\times m)^2\delta(m+n),
  \]
\[
(b_m,b_n)=(\beta\times m)^2\delta(m+n),
 \] 
\[(a_m,b_n)=(\alpha\times m).(\beta\times
n)\delta(m+n),
 \] 
 \[
 (c_j,c_k)=\delta_{jk}=(d_j,d_k).
 \]
 
 The
fluid flows along geodesics determined by this metric along with the commutation
relations(\ref{eq:aabb}).

$\hat{\mathcal{S}}$ is not a co-boundary Lie bi-algebra; i.e., there is no classical $r$-matrix such that $\mu=\partial r$ in the Lie algebra co-homology of $\mathcal{A}$. This makes it harder to construct the quantum group: it is not determined by an $R$-matrix.

I thank  Mark Bowick and Jaemo Park  for discussions. Some of this work was done during a visit to the KITP (Santa Barbara) whose research was supported in part by the National Science Foundation under Grant No. NSF PHY-1748958.

\vfill\pagebreak

\section{ Details of Proofs}

\subsection{The Central Extension }
\begin{prop}
Each component of $\Omega(u,w)=\int w\times u \dx$ is a 2-cocycle of
$\mathcal{S}$. i.e.,$\partial\Omega(u,v,w)\equiv\Omega(u,[v,w])+\Omega(v,[w,u])+\Omega(w,[u,v])=0$
\end{prop}

\begin{proof}
\begin{equation}
-\partial\Omega(u,v,w)=\int u\times\left(v^{i}\partial_{i}w-w^{i}\partial_{i}v\right)\dx+\mathrm{cyclic}
\end{equation}

Because of incompressibility,

\begin{equation}
=\int u\times\partial_{i}\left(v^{i}w-w^{i}v\right)\dx+\mathrm{cyclic}
\end{equation}

Integrating by parts 

\begin{equation}
=-\int\partial_{i}u\times\left(v^{i}w-w^{i}v\right)\dx+\mathrm{cyclic}
\end{equation}

\begin{equation}
=\int\left[w^{i}\partial_{i}u\times v-v^{i}\partial_{i}u\times w\right]\dx+\mathrm{cyclic}
\end{equation}

Cyclically permuting the first term

\begin{equation}
=\int\left[u^{i}\partial_{i}v\times w-v^{i}\partial_{i}u\times w\right]\dx+\mathrm{cyclic}
\end{equation}

\begin{equation}
=\int\left[u,v\right]\times w\ \dx+\mathrm{cyclic}
\end{equation}

anti-symmetry of cross product

\begin{equation}
=-\int w\times[u,v]\dx+\mathrm{cyclic}
\end{equation}

Cyclic permutation

\begin{equation}
=-\int u\times[v,w]\dx+\mathrm{cyclic}
\end{equation}

\begin{equation}
=\partial\Omega(u,v,w)
\end{equation}

Thus $\partial\Omega=0$.
\end{proof}

This allows us to define a central extension. The vector space 

\begin{equation}
\hat{\mathcal{S}}=\mathbb{R}^{3}\oplus\mathcal{S}
\end{equation}

is turned into a Lie algebra by 

\begin{equation}
\left[(u,\eta),(w,\mu)\right]=\left([u,w],\Omega(u,w)\right)
\end{equation}

The Jacobi identity follows from the co-cycle condition proved above.

\begin{prop}
The Lie algebra $\hat{\mathcal{S}}$ admits the invariant inner product
$\langle(u,\eta),(w,\mu)\rangle=\int(u-\bar{u}).\mathrm{curl}^{-1}(w-\bar{w})\dx+\eta.\bar{w}+\mu.\bar{u}$,
where $\bar{u}=\int u\dx$.
\end{prop}

\begin{proof}
We have
\begin{equation}
\langle(u,\eta),\left[(w,\mu),\left(v,\sigma\right)\right]\rangle=\langle(u,\eta),\left([w,v],\Omega(w,v)\right)\rangle
\end{equation}

Since $\overline{[w,v]}=0$ and $[w,v]=\mathrm{curl}(v\times w)$

\begin{equation}
=\int(u-\bar{u}).(v\times w)dx+\bar{u}.\int v\times w \dx
\end{equation}

\begin{equation}
=\int u.(v\times w)\dx.
\end{equation}

The anti-symmetry of the triple scalar product now proves the invariance
of the inner product, as before.
\end{proof}

Thus the central basis elements are dual to translations under the
invariant inner product.

\subsection{Fourier Basis}
\begin{lem}
For $m\neq0$, we can expand any vector as $w=\frac{m\cdot w}{m\cdot m}m+\frac{m\cdot w\ \beta.m-m\cdot m\beta.w}{m\cdot m\ \gamma\cdot m}\alpha\times m+\frac{-m\cdot w\ \alpha.m+m\cdot m\alpha.w}{m\cdot m\ \gamma\cdot m}\beta\times w$
\end{lem}

\begin{proof}
Note that $m,\alpha\times m,\beta\times m$ form a basis because 

$\det\left(\begin{array}{ccc}
m_{1} & m_{3}\alpha_{2}-m_{2}\alpha_{3} & m_{3}\beta_{2}-m_{2}\beta_{3}\\
m_{2} & m_{1}\alpha_{3}-m_{3}\alpha_{1} & m_{1}\beta_{3}-m_{3}\beta_{1}\\
m_{3} & m_{2}\alpha_{1}-m_{1}\alpha_{2} & m_{2}\beta_{1}-m_{1}\beta_{2}
\end{array}\right)=m\cdot m(\alpha\times\beta)\cdot m\neq0$. Expand $w=\xi_{1}m+\xi_{2}\alpha\times m+\xi_{3}\beta\times m$.
Taking scalar products with $m,\beta,\alpha$ we get 

\begin{equation}
m\cdot w=m.m\xi_{1},\quad\alpha.w=\alpha.m\xi_{1}+\gamma\cdot m\xi_{3},\quad\beta.w=\beta.m\xi_{1}-\gamma.m\xi_{2}
\end{equation}

Solving, we get 

\begin{equation}
\xi_{1}=\frac{m\cdot w}{m\cdot m},\quad\xi_{2}=\frac{m\cdot w\ \beta.m-m\cdot m\beta.w}{m\cdot m\ \gamma\cdot m},\quad\xi_{3}=\frac{-m\cdot w\ \alpha.m+m\cdot m\alpha.w}{m\cdot m\ \gamma\cdot m}
\end{equation}
\end{proof}
In particular, if $m.w=0$ 

\begin{equation}
w=-\frac{\beta.w}{\gamma\cdot m}\alpha\times m,+\frac{\alpha.w}{\gamma\cdot m}\beta\times m.
\end{equation}

\begin{prop}
Any incompressible trigonometric polynomial vector field can be written
as $u=\bar{u}+\alpha\times\nabla A+\beta\times\nabla B$ where $\bar{u}$
is a constant vector field and $A,B\in\mathcal{F}$. Moreover, $A$
and $B$ are unique up to additive constants.
\end{prop}

\begin{proof}
Expand in a Fourier series

\begin{equation}
u(x)=\bar{u}+\sum_{m\in\mathbb{Z}'^{3}}u_{m}e^{im\cdot x},\quad m\cdot u_{m}=0
\end{equation}

where $\bar{u}$ is a constant vector. Using 
\begin{equation}
u_{m}=-\alpha\times m\ \frac{\beta\cdot u_{m}}{\gamma\cdot m}+\beta\times m\ \frac{\alpha\cdot u_{m}}{\gamma\cdot m}
\end{equation}

\begin{equation}
a_{m}=-i\alpha\times me_{m},\quad b_{m}=-i\beta\times me_{m}
\end{equation}
we get the expansion

\begin{equation}
u=u_{0}+\sum_{m\in\mathbb{Z}'^{3}}\left[-i\frac{\beta\cdot u_{m}}{\gamma\cdot m}a_{m}+i\frac{\alpha\cdot u_{m}}{\gamma\cdot m}b_{m}\right]
\end{equation}

Finally, since $\gamma\equiv\alpha\times\beta\neq0$, it is obvious
that $\alpha,\beta,\gamma$ form a basis for the constant vectors.

We can also write this as 

\begin{equation}
u=\bar{u}+\alpha\times\nabla A+\beta\times\nabla B
\end{equation}

where the scalar fields $A,B$ are given by 

\begin{equation}
A=(\gamma\cdot\nabla)^{-1}\beta\cdot(u-\bar{u}),\quad B=-(\gamma\cdot\nabla)^{-1}\alpha\cdot(u-\bar{u})
\end{equation}

Also, $\bar{u}$ is the average over the torus of $u$.
\end{proof}
The above decomposition shows that 
\begin{cor}
A constant vector field is incompressible, but is not the curl of
any vector field on the torus (i.e., is not exact). An incompressible
vector field whose average is zero can be written as the curl of another
such vector field.
\end{cor}

\begin{proof}
Fourier analysis shows that the average of $\nabla\times U$ for any
$U$ with periodic components is zero: the constant terms in $U$
have zero curl. 
\end{proof}
\begin{prop*}
In $\mathcal{S}$, we have the relations
\end{prop*}
\begin{equation}
[a_{m},a_{n}]=\alpha\cdot(m\times n)a_{m+n},\quad[b_{m},b_{n}]=\beta\cdot(m\times n)b_{m+n}\label{eq:aabb-1-1}
\end{equation}
\begin{equation}
[a_{m},b_{n}]=\frac{\gamma.m}{\gamma\cdot(m+n)}\beta\cdot(m\times n)a_{m+n}+\frac{\gamma\cdot n}{\gamma\cdot(m+n)}\alpha\cdot(m\times n)b_{m+n}\label{ab-1}
\end{equation}

It is straightforward to verify (\ref{eq:aabb-1-1}). To prove (\ref{ab-1})
we need the lemma
\begin{lem*}
When $\gamma=\alpha\times\beta$, we have the identity 
\begin{equation}
\alpha.(m\times n)\left\{ \gamma.n\beta\times m-\gamma.m\beta\times n\right\} -\beta.(m\times n)\left\{ \gamma.n\alpha\times m-\gamma.m\alpha\times n\right\} =0
\end{equation}
\end{lem*}
\begin{proof}
Since $\alpha\times\beta=\gamma$ and $\alpha.(\beta\times m)=(\alpha\times\beta).m=\gamma.m$
we have 

\begin{equation}
\alpha.\left\{ \gamma.n\beta\times m-\gamma.m\beta\times n\right\} =0,\quad\beta.\left\{ \gamma.n\alpha\times m-\gamma.m\alpha\times n\right\} =0
\end{equation}

Of course also $\beta.(\beta\times m)=0$ so that 

\begin{equation}
\beta.\left\{ \gamma.n\beta\times m-\gamma.m\beta\times n\right\} =0,\quad\alpha.\left\{ \gamma.n\alpha\times m-\gamma.m\alpha\times n\right\} =0
\end{equation}

Since $\left\{ \alpha,\beta,\gamma\right\} $ is a basis, and we have
proved that the $\alpha$ and $\beta$ components are zero, it is
enough to prove that the $\gamma-$ component is zero as well. Now
recall that 

\begin{equation}
u.mw.n-w.m\ u.n=\left(u\times w\right).\left(m\times n\right)
\end{equation}

so that 

\begin{equation}
\gamma.\left\{ \gamma.n\beta\times m-\gamma.m\beta\times n\right\} =\gamma.n(\gamma\times\beta).m-\gamma.m(\gamma\times\beta).n
\end{equation}

\begin{equation}
=-(\gamma\times(\gamma\times\beta)).(m\times n)
\end{equation}

\begin{equation}
=\gamma.\gamma\ \beta.(m\times n)
\end{equation}

Similarly

\begin{equation}
\gamma.\left\{ \gamma.n\alpha\times m-\gamma.m\alpha\times n\right\} =\gamma.\gamma\alpha.(m\times n)
\end{equation}

Thus

\begin{equation}
\gamma.\left[\alpha.(m\times n)\left\{ \gamma.n\beta\times m-\gamma.m\beta\times n\right\} -\beta.(m\times n)\left\{ \gamma.n\alpha\times m-\gamma.m\alpha\times n\right\} \right]=
\end{equation}

\begin{equation}
=\gamma.\gamma\left[\alpha.(m\times n)\beta.(m\times n)-\beta.(m\times n)\alpha.(m\times n)\right]=0
\end{equation}

as needed.
\end{proof}
Now we can prove (\ref{ab-1})
\begin{proof}
Start with 

\begin{equation}
a_{m}=-i\alpha\times me_{m},\quad b_{n}=-i\beta\times ne_{n}
\end{equation}

Then 

\begin{equation}
a_{m}.\nabla b_{n}=-i\left(\alpha\times m\right).n\ \beta\times n\ e_{m+n}
\end{equation}

\begin{equation}
=\frac{\gamma\cdot n}{\gamma\cdot(m+n)}\alpha\cdot(m\times n)b_{m+n}+\frac{\alpha.(m\times n)}{\gamma.(m+n)}\left\{ \gamma.n\beta\times(m+n)-\beta\times n\gamma.(m+n)\right\} ie_{m+n}
\end{equation}

\begin{equation}
a_{m}.\nabla b_{n}=\frac{\gamma\cdot n}{\gamma\cdot(m+n)}\alpha\cdot(m\times n)b_{m+n}+\frac{\alpha.(m\times n)}{\gamma.(m+n)}\left\{ \gamma.n\beta\times m-\gamma.m\beta\times n\right\} ie_{m+n}
\end{equation}

Similarly

\begin{equation}
b_{n}.\nabla a_{n}=\frac{\gamma\cdot m}{\gamma\cdot(m+n)}\beta\cdot(n\times m)a_{m+n}+\frac{\beta.(n\times m)}{\gamma.(m+n)}\left\{ \gamma.m\alpha\times n-\gamma.n\alpha\times m\right\} ie_{m+n}
\end{equation}

so that 

\begin{equation}
[a_{m},b_{n}]=\frac{\gamma\cdot n}{\gamma\cdot(m+n)}\alpha\cdot(m\times n)b_{m+n}+\frac{\gamma\cdot m}{\gamma\cdot(m+n)}\beta\cdot(m\times n)a_{m+n}
\end{equation}
\begin{equation}
+\left[\alpha.(m\times n)\left\{ \gamma.n\beta\times m-\gamma.m\beta\times n\right\} -\beta.(m\times n)\left\{ \gamma.n\alpha\times m-\gamma.m\alpha\times n\right\} \right]\frac{ie_{m+n}}{\gamma.(m+n)}
\end{equation}
The last term is zero by the Lemma above.
\end{proof}

\subsection{Symmetric Version of the Commutation relations}

\begin{equation}
[X^{m},X^{n}]=\frac{1}{(i\gamma.m)(i\gamma.n)}[b_{-m},b_{-n}]=\frac{\beta.(m\times n)}{(i\gamma.m)(i\gamma.n)}b_{-m-n}
\end{equation}
\begin{equation}
=\frac{\beta.(m\times n)(-i\gamma.(m+n))}{(i\gamma.m)(i\gamma.n)}X^{m+n}
\end{equation}
Thus
\begin{equation}
\mu_{k}^{mn}\equiv\delta_{k}^{m+n}\frac{(-i\gamma.k)}{(i\gamma.m)(i\gamma.n)}\beta.(m\times n)
\end{equation}
Finally,

\begin{equation}
[X_{m},X^{n}]=-\frac{1}{i\gamma.n}[a_{m},b_{-n}]
\end{equation}

\begin{equation}
=-\frac{\gamma.m}{i\gamma.n\gamma\cdot(m-n)}\beta\cdot(m\times n)a_{m-n}+\frac{\gamma\cdot n}{i\gamma.n\gamma\cdot(m-n)}\alpha\cdot(m\times n)b_{m-n}
\end{equation}
\begin{equation}
=-\frac{\gamma.m}{i\gamma.n\gamma\cdot(m-n)}\beta\cdot(m\times n)a_{m-n}+\alpha\cdot(m\times n)X^{n-m}
\end{equation}

Now

\begin{equation}
\Gamma_{mk}^{n}X^{k}=\delta_{k+m}^{n}\alpha.(m\times k)X^{k}=\alpha.(m\times n)X^{n-m}
\end{equation}
\begin{equation}
\mu_{m}^{nk}X_{k}=\delta_{m}^{n+k}\frac{(-i\gamma.m)}{i\gamma.n\ i\gamma.k}\beta.(n\times k)X_{k}=\frac{(-i\gamma.m)}{i\gamma.n\ i\gamma.(m-n)}\beta.(n\times m)X_{m-n}
\end{equation}
\begin{equation}
=-\frac{(-i\gamma.m)}{i\gamma.n\ i\gamma.(m-n)}\beta.(m\times n)X_{m-n}
\end{equation}
\begin{equation}
=\frac{\gamma.m}{i\gamma.n\ \gamma.(m-n)}\beta.(m\times n)X_{m-n}
\end{equation}
Thus
\begin{equation}
[X_{m},X^{n}]=\Gamma_{mk}^{n}X^{k}-\mu_{m}^{nk}X_{k}
\end{equation}

\subsection{Proof that $\mu$ is not a co-boundary}

The structure constants $\Gamma,\mu$ of a Lie bi-algebra satisfy the identities 
\[
\Gamma_{ab}^{c}=-\Gamma_{ba}^{c},\quad\Gamma_{ab}^{d}\Gamma_{dc}^{e}+\Gamma_{bc}^{d}\Gamma_{da}^{e}+\Gamma_{ca}^{d}\Gamma_{db}^{e}=0
\]
\[
\mu_{d}^{bc}\Gamma_{ae}^{d}=\left[\Gamma_{ad}^{b}\mu_{e}^{dc}+\mu_{a}^{bd}\Gamma_{de}^{c}\right]-b\leftrightarrow c
\]
which are the just Jacobi identities for the Lie sub-algebras $\mathcal{A},\mathcal{B}$ spanned by $X_a$ and $X^a$ respectively. The mixed Jacobi identitities $[X_a,X_b,X^c]]$ and $[X^a,X^b,X_c]$ both lead to the condition
\[
\mu_{d}^{bc}\Gamma_{ae}^{d}=\left[\Gamma_{ad}^{b}\mu_{e}^{dc}+\mu_{a}^{bd}\Gamma_{de}^{c}\right]-b\leftrightarrow c
\]
This has another meaning: it says that $\mu$ is a co-cycle in the Lie algebra cohomogy $H^1(\mathcal{A},\mathcal{A}\otimes \mathcal{A})$ (or conversely, $\Gamma$ co-cycle in $H^1(\mathcal{B},\mathcal{B}\otimes\mathcal{B}))$.

If $\mathcal{A}$ is a finite dimensional simple algebra (such as $sl_2$) this is a co-boundary; i.e., there is an $r\in \mathcal{A}\otimes \mathcal{A}$  such that $\mu=\partial r $. The Jacobi identity of $\mu$ then becomes  a quadratic condition on $r$ called the classical Yang-Baxter equation. This is the infinitesimal version of the famous Yang-Baxter equation of a quasi-triangular quantum group, which is the exponential of such a co-boundary Lie bi-algebra.

It would have  been a simplification if such an $r-$ matrix existed for our Lie bi-algebra; being infinite dimensional, the usual arguments for its existence do not apply. A direct study is needed. We will now that 

\begin{prop}
For the Lie bi-algebra $\hat{\mathcal{S}}$ the structure constants $\mu$ are not a co-boundary; there is no classical $r-$ matrix arising from it.
\end{prop}

\begin{proof}
We have 
\[
\mu(X_{k})=\sum_{m+n=k}\frac{(-i\gamma.k)}{(i\gamma.m)(i\gamma.n)}\beta.(m\times n)X_{m}\otimes X_{n}
\]

Let us calculate the co-boundary of an element $\rho=\rho^{pq}X_{p}\otimes X_{q}\in\mathcal{A}\otimes\mathcal{A}$. 

\[
\partial\rho(X_{k})=\left[X_{k}\otimes1+1\otimes X_{k},\rho\right]
\]

\[
=\left[X_{k}\otimes1+1\otimes X_{k},\rho^{pq}X_{p}\otimes X_{q}\right]
\]

\[
=\sum_{pq}\rho^{pq}\left\{ \left[X_{k},X_{p}\right]\otimes X_{q}+X_{p}\otimes[X_{k},X_{q}]\right\} 
\]

\[
=\sum_{pq}\rho^{pq}\left\{ \alpha.(k\times p)X_{k+p}\otimes X_{q}+\alpha.(k\times q)X_{p}\otimes X_{k+q}\right\} 
\]

\[
=\sum_{pq}\rho^{pq}\left\{ \alpha.(k\times p)X_{k+p}\otimes X_{q}+\alpha.(k\times q)X_{p}\otimes X_{k+q}\right\} 
\]

Replace $p\mapsto p-k,q\mapsto q+k$ in the first term 

\[
=\sum_{pq}\left\{ \rho^{p-k,q+k}\alpha.(k\times p)X_{p}\otimes X_{q+k}+\rho^{pq}\alpha.(k\times q)X_{p}\otimes X_{k+q}\right\} 
\]

\[
=\sum_{pq}\left\{ \rho^{p-k,q+k}\alpha.(k\times p)+\rho^{pq}\alpha.(k\times q)\right\} X_{p}\otimes X_{q+k}
\]

Put $m=p,n=q+k$

\[
\partial\rho(X_{k})=\sum_{mn}\left\{ \rho^{m-k,n}\alpha.(k\times m)+\rho^{m,n-k}\alpha.(k\times n)\right\} X_{m}\otimes X_{n}
\]

Compare with 
\[
\mu(X_{k})=\sum_{m,n}\delta_{k}^{m+n}\frac{(-i\gamma.k)}{(i\gamma.m)(i\gamma.n)}\beta.(m\times n)X_{m}\otimes X_{n}
\]

\[
\delta_{k}^{m+n}\frac{(-i\gamma.k)}{(i\gamma.m)(i\gamma.n)}\beta.(m\times n)=\rho^{m-k,n}\alpha.(k\times m)+\rho^{m,n-k}\alpha.(k\times n)
\]

So rhs must vanish if $k\neq m+n$. This suggests the ansatz 

\[
\rho^{mn}=\delta(m+n)\rho^{m}
\]

\[
\delta_{k}^{m+n}\frac{(-i\gamma.k)}{(i\gamma.m)(i\gamma.n)}\beta.(m\times n)=\left\{ -\rho^{-n}\alpha.(n\times m)-\rho^{m}\alpha.(m\times n)\right\} \delta(k+m+n)
\]

\[
\frac{(-i\gamma.(m+n))}{(i\gamma.m)(i\gamma.n)}\beta.(m\times n)=\left\{ \rho^{m}+\rho^{n}\right\} \alpha.(m\times n)
\]

Let 

\[
M_{mn}=\frac{(-i\gamma.(m+n))}{(i\gamma.m)(i\gamma.n)}\frac{\beta.(m\times n)}{\alpha.(m\times n)}
\]

If $\mu$ is a co-boundary,  this should be equal to $\sigma^{mn}\equiv \rho^{m}+\rho^{n}$. But this is impossible. 

For,   $\sigma^{mn}$ is a rank two matrix, being of  the form $\rho\otimes \xi+\xi\otimes \rho $ where $\xi$ is the vector all of whose components are equal to one. 
It is easy to check that $M$ has sub-matrices of rank higher than two. For example, 
we can very directly that (with the choice $\alpha=(1,\sqrt{2},\sqrt{3}),\beta=(\sqrt{3},1,\sqrt{2})$ given in the text )  the sub-matrix $M_1\subset M$  labeled by

\[
m=\{(3,2,3),(4,3,4),(4,1,1)\},\quad n=\{(4, 3, 2), (3, 4,2), (2, 4, 3)\}
\]
\[
M_1=i\left(
\begin{array}{ccc}
 \frac{-79 \sqrt{2}+222 \sqrt{3}+12 \sqrt{6}-299}{-2149 \sqrt{2}+16 \sqrt{3}+225 \sqrt{6}+2466} & \frac{312 \sqrt{2}-358 \sqrt{3}-3 \sqrt{6}+177}{281 \sqrt{2}+2426 \sqrt{3}+27 \sqrt{6}-4676} & \frac{398 \sqrt{2}-285 \sqrt{3}+14 \sqrt{6}-108}{-2865 \sqrt{2}+3806 \sqrt{3}+245 \sqrt{6}-3144} \\
 \frac{2 \left(-15 \sqrt{2}+76 \sqrt{3}+3 \sqrt{6}-116\right)}{-2196 \sqrt{2}+380 \sqrt{3}+214 \sqrt{6}+1923} & \frac{-439 \sqrt{2}+522 \sqrt{3}+3 \sqrt{6}-276}{-915 \sqrt{2}-3910 \sqrt{3}+14 \sqrt{6}+8031} & \frac{575 \sqrt{2}-396 \sqrt{3}+17 \sqrt{6}-174}{-5019 \sqrt{2}+6602 \sqrt{3}+281 \sqrt{6}-5025} \\
 \frac{-273 \sqrt{2}+47 \sqrt{3}+48 \sqrt{6}+164}{1805 \sqrt{2}-2682 \sqrt{3}-331 \sqrt{6}+2878} & \frac{-459 \sqrt{2}+107 \sqrt{3}+66 \sqrt{6}+248}{2385 \sqrt{2}-4306 \sqrt{3}-115 \sqrt{6}+4298} & \frac{-496 \sqrt{2}+83 \sqrt{3}+39 \sqrt{6}+432}{4309 \sqrt{2}-4718 \sqrt{3}-547 \sqrt{6}+3390} \\
\end{array}
\right)
\]
is of rank $3$. In fact, we expect that $M$ is of infinite rank.

\end{proof}

\end{document}